\documentclass[11pt,a4,reqno,tbtags]{amsart}
\usepackage{amsmath}
\usepackage{amsfonts}
\usepackage{amssymb}
\usepackage{latexsym}
\usepackage{graphicx}
\usepackage[utf8]{inputenc}
\usepackage{mathrsfs}
\usepackage{epstopdf}
\usepackage{hyperref}
\usepackage{url}

\allowdisplaybreaks
\numberwithin{equation}{section}

\newtheorem{proposition}{Proposition}[section]

\newtheorem{theorem}[proposition]{Theorem}

\newtheorem{lemma}[proposition]{Lemma}

\theoremstyle{definition}
\newtheorem{remark}[proposition]{Remark}

\newtheorem*{ack}{Acknowledgement}

\theoremstyle{remark}

\renewcommand\P{\mathbb{P}}

\newcommand{\bea}{\begin{eqnarray}}
\newcommand{\eea}{\end{eqnarray}}

\def\void{}
\def\labelmark{}

\newenvironment{formula}[1]{\def\labelname{#1}
\ifx\void\labelname\def\junk{\begin{displaymath}}
\else\def\junk{\begin{equation}\label{\labelname}}\fi\junk}%
{\ifx\void\labelname\def\junk{\end{displaymath}}
\else\def\junk{\end{equation}}\fi\junk\labelmark\def\labelname{}}

{\ifx\void\labelname\def\junk{\end{array}\end{displaymath}}
\else\def\junk{\end{array}\right.\end{equation}}
\fi\junk\labelmark\def\labelname{}\def\junk{}
}

\newcommand{\beq}{\begin{formula}}
\newcommand{\eeq}{\end{formula}}
\newcommand{\beqv}{\begin{formula}{}}

\newcommand\urladdrx[1]{{\urladdr{\def~{{\tiny$\sim$}}#1}}}

\newenvironment{romenumerate}[1][0pt]{
\addtolength{\leftmargini}{#1}\begin{enumerate}
 }{\end{enumerate}}

\newcounter{oldenumi}
\newenvironment{romenumerateq}
{\setcounter{oldenumi}{\value{enumi}}
\begin{romenumerate} \setcounter{enumi}{\value{oldenumi}}}
{\end{romenumerate}}

\begingroup
  \divide\count255 by 60
  \multiply\count255 by -60
  \advance\count255 by \time
  \ifnum \count255 < 10 \xdef\klockan{\the\count1.0\the\count255}
  \else\xdef\klockan{\the\count1.\the\count255}\fi
\endgroup

\newcommand\REM[1]{{\raggedright\texttt{[#1]}\par\marginal{XXX}}}

\newcommand\ga{\alpha}
\newcommand\gb{\beta}
\newcommand\gd{\delta}

\newcommand\gG{\Gamma}

\newcommand\gl{\lambda}

\newcommand\gs{\sigma}

\newcommand\eps{\varepsilon}

\newcommand{\cN}{\mathcal{N}}

\newcommand\set[1]{\ensuremath{\{#1\}}}

\newcommand\bigpar[1]{\bigl(#1\bigr)}
\newcommand\Bigpar[1]{\Bigl(#1\Bigr)}

\newcommand\lrpar[1]{\left(#1\right)}

\newcommand\lrcpar[1]{\left\{#1\right\}}

\newcommand\Bigabs[1]{\Bigl|#1\Bigr|}

\def\rompar(#1){\textup(#1\textup)}    

\newcommand\parfrac[2]{\lrpar{\frac{#1}{#2}}}

\newcommand\Bigparfrac[2]{\Bigpar{\frac{#1}{#2}}}

\def\xexp(#1){e^{#1}}
\newcommand\ceil[1]{\lceil#1\rceil}
\newcommand\floor[1]{\lfloor#1\rfloor}

\newcommand\Ntoo{\ensuremath{{N\to\infty}}}

\newcommand\upto{\nearrow}

\newcommand{\tend}{\longrightarrow}
\newcommand\dto{\overset{\mathrm{d}}{\tend}}
\newcommand\pto{\overset{\mathrm{p}}{\tend}}

\newcommand\bbN{\mathbb N}

\newcounter{CC}
\newcommand{\CC}{\stepcounter{CC}\CCx} 
\newcommand{\CCx}{C_{\arabic{CC}}}     
\newcommand{\CCdef}[1]{\xdef#1{\CCx}}     
\newcounter{cc}

\newcommand\sumn{\sum_{n=0}^\infty}
\newcommand\sumik{\sum_{i=1}^K}
\newcommand\prodik{\prod_{i=1}^K}
\newcommand\dx{D_*}
\newcommand\neps{L} 
\newcommand\Pois{\operatorname{Pois}}
\newcommand\tD{\widetilde D}
\newcommand\etto{\bigpar{1+o(1)}}
\newcommand\nx{n^*}
\newcommand\nnx{\nx_1,\dots,\nx_K}
\newcommand\nnxx{\mathbf{n^*}}
\newcommand\nnq{n_1,\dots,n_K}
\newcommand\nnqq{\mathbf{n}}
\newcommand\mm{m_1,\dots,m_K}
\newcommand\mmqq{\mathbf{m}}
\newcommand\gdij{\gd_{ij}}

\begin{document}
\title
{Random trees with superexponential branching weights}

\date{14 May, 2011}   
\subjclass[2000]{05C80, 05C05, 60J80, 60F05} 
\keywords{Random trees, simply generated trees, branching process, weak limit.}

\author{Svante Janson}
\address{Department of Mathematics, Uppsala University, PO Box 480,
SE-751~06 Uppsala, Sweden}
\email{svante.janson@math.uu.se}
\urladdrx{http://www2.math.uu.se/~svante/}

\author{Thordur Jonsson}  
\address{The Science Institute, University of Iceland,
Dunhaga 3, 107 Reykjavik, Iceland}
\email{thjons@raunvis.hi.is}

\author{Sigurdur \"Orn Stef\'ansson} 
\address{NORDITA,
Roslagstullsbacken 23, SE-106 91 Stockholm, Sweden}
\email{sigste@nordita.org}



\begin{abstract} 
We study rooted planar random trees with a probability distribution
which is proportional to a product of weight factors $w_n$ associated
to the vertices of the tree and depending only on their individual
degrees $n$.  We focus on the case when $w_n$ grows faster than
exponentially with $n$.  In this case the measures on trees of finite
size $N$ converge weakly as $N$ tends to infinity to a measure which
is concentrated on a single tree with one vertex of infinite degree.
For explicit weight factors of the form $w_n=((n-1)!)^\alpha$ with
$\alpha >0$ we obtain more refined results about the approach to
the infinite volume limit.
\end{abstract}

\maketitle

\section{Introduction}
Random trees have been studied intensively by mathematicians and
theoretical physicists in the last few decades.  They have direct 
applications to many branches in science, they are essential in many
mathematical models used by physicists and are a natural object to
study from the point of view of pure mathematics.  

The random trees we are concerned with here were originally called {\it
simply generated trees} by probabilists \cite{MeirMoon}.  Later the
same tree ensembles were referred to as {\it random trees with a local
action} by physicists and viewed as toy models in statistical
mechanics and for some aspects of quantum gravity, see e.g.\ \cite{adj}.

Simply generated trees with $N$ vertices can be defined as follows:
Let $(w_n)_{n\geq 1}$ be a sequence of nonnegative numbers which we
will call {\it branching weights}.
If $T$ is a tree graph with vertex set $V(T)$ having $N$ elements we define a
probability distribution on the set of all such trees by the formula
\begin{equation}
\nu (T)=Z^{-1}\prod_{v\in V(T)}w_{\sigma (v)},
\end{equation}
where $\sigma (v)$ is the degree of the vertex $v$ and $Z$ is a
normalization factor called partition function in physics.  One is
interested in typical properties of trees with respect to this
measure, especially asymptotics for large $N$ and the existence of a
limiting measure as $N\to\infty$.  

A lot is known about such trees for ``nice'' branching weights as we
review briefly below.  In this
paper we aim at complementing some of these results for weights $w_n$
which grow faster than exponentially with $n$.  In this case some of
the formalism that has been used to study simply generated trees is
not applicable any more as we will explain below.  A physicist 
would say that the Grand partition function is divergent which normally 
is a signal of instability in a
physical theory.  We will indeed see that with superexponential
branching weights one vertex becomes connected to all the other
vertices in the infinite volume limit.

In the next section we give a more technical background and summarize
our results.  The final section contains detailed proofs.

\section{Definitions and summary of results}

We consider rooted planar trees with root $r$ of degree 1.  We let
$\Gamma_N$ be the set of trees with $N$ edges and denote the set of finite
and infinite trees by $\Gamma$.  Vertices of infinite order are
allowed and for such vertices the links pointing away from the root
are ordered as $\bbN$, i.e.\ there is a leftmost edge pointing away
from the root.  The unique nearest neighbour of the root $r$ will be denoted by
$s$. 

\begin{remark}\label{Rroot}
  We include the root $r$ just for convenience. It is equivalent to omit it
  and consider $s$ as the root (now with arbitrary degree), with minor
  changes in the notation; $N$ is then the number of vertices in the tree
  and the degree $\gs(v)$ is replaced by $1+\gs_+(v)$ where $\gs_+(v)$ is
  the outdegree of $v$. It may be even more convenient to omit $r$ but keep
  the pendant edge from $s$ to $r$ as an edge with one free endpoint; this
  point of view is used 
  sometimes in the proofs below.
\end{remark}

\begin{remark}\label{RUlamHarris}
We can regard the set $\gG$ as a set of subtrees of the infinite Ulam--Harris
tree $T_\infty$, which is the tree with vertex set 
$V(T_\infty)=\set{r}\cup\bigcup_{k=0}^\infty\bbN^k$, the set of all finite
strings of natural numbers (and $r$), with $s=\emptyset$ (the empty string,
so $\bbN^0=\set s$) 
and a vertex $v=v_1\cdots v_k$ having
ancestor 
$v_1\cdots v_{k-1}$ when $k>0$.
More precisely, $\gG$ can be identified with the set of all rooted subtrees 
$T$ of
$T_\infty$ such that if $v=v_1\cdots v_k$ is a vertex in $T$, then so is
$v_1\cdots v_{k-1}i$ for every $i<v_k$. We call such subtrees of $T_\infty$ {\it left subtrees} and more generally, we say that a tree $T' \in \Gamma$ is a left subtree of $T \in \Gamma$ if $V(T') \subseteq V(T)$.
\end{remark}

We endow $\Gamma$ with a metric $d$ which 
is defined as follows: Let $T\in \Gamma$ and define $B_R(T)$ as the graph ball of radius $R$, centered on the root $r$ in $T$. The {\it left ball} of radius $R$, $L_R(T)$, is defined as the maximal left subtree of $B_R(T)$ with vertices of degree no greater than $R$. The metric $d$ is given by
\begin {equation} \label{defmetric}
 d(T,T') = \inf\left\{ \frac{1}{R+1}~\Big|~ L_R(T) = L_R(T') \right\}, \quad T,T'\in \Gamma.
\end {equation}
Convergence in $\gG$, in the metric $d$, is equivalent to
convergence of the degree $\gs(v)$ for every $v\in V(T_\infty)$ (where we
define $\gs(v)=0$ for $v\notin T$), see \cite{sdf} for details.

To avoid trivialities we assume that the branching weights satisfy 
$w_1\neq0$ and $w_n\neq0$ for at least some $n>2$. 
We define the finite volume partition function
\begin {equation}
Z_N = \sum_{\tau\in\Gamma_N} ~\prod_{v\in V(\tau)\setminus\{r\}} w_{\sigma(v)}
\end {equation}
and a probability distribution $\nu_N$ on $\Gamma_N$ by
\begin {equation}
 \nu_N(\tau) = Z_N^{-1} \prod_{v\in V(\tau)\setminus\{r\}} w_{\sigma(v)}.
\end {equation}
This probability distribution describes a random tree $T_N$ with $N$ edges.

Let $\rho\ge0$ be the radius of convergence of 
the generating function 
\begin{equation}
g(z)=\sumn w_{n+1}z^n  
\end{equation}
of the branching weights.
A rescaling $w_n\mapsto ab^nw_n$ with $a,b>0$ does not affect the
distributions $\nu_N$, and it is well-known and easy to see that if
$\rho>0$, we can by rescaling assume that $(w_n)$ is a probability
distribution, i.e.\ $\sum_0^\infty w_n=1$. In that case, the random tree $T_N$
with distribution $\nu_N$ is a Galton--Watson tree with offspring
distribution $(w_{n+1})_{n=0}^\infty$, conditioned to have size $N$. If
further $\lim_{z\upto\rho} zg'(z)/g(z)\ge1$,
then the distributions $\nu_N$ converge to the distribution of a random tree
that is infinite, with all vertex degrees finite and exactly one infinite
path, see further \cite{durhuus,djw,LPP}.  The limiting measure
describes an infinite critical Galton--Watson tree conditioned on
nonextinction.
On the other hand, in the subcritical case when
$m=\lim_{z\upto\rho} zg'(z)/g(z)<1$, then (at least under some technical
conditions) the limit distribution still exists but now describes
a random tree with exactly one vertex of infinite degree; the 
length of the path from $r$ to this vertex has a geometric distribution with
mean $1/(1-m)$; the rest of the tree can be described by 
a subcritical Galton--Watson process, see \cite{sdf} for details.

In the present paper we are interested in the case when the radius of
convergence $\rho=0$. Note that then there is no Galton--Watson
interpretation.
We prove in Section \ref{Spf} weak convergence, as $N\rightarrow \infty$,
of the measures $\nu_N$ 
(in the topology generated by $d$) in this case too,
under certain conditions on the weights. 
The result can be seen as a natural limiting case of the result in
\cite{sdf} as $m\to0$; the resulting limit tree is in this case
non-random, and is simply an infinite star.

\begin {theorem}\label{th1}
 If the branching weights satisfy 
\begin {equation} \label{condition1}
 \frac{w_{n+1}}{w_n} \xrightarrow[n\rightarrow \infty]{} \infty
\end {equation}
then the measures $\nu_N$ viewed as probability measures on $\Gamma$,
converge weakly to the probability measure that is concentrated on the 
single tree which has $\sigma(s) = \infty$ and all other vertices of degree one.
\end {theorem}

Furthermore, we obtain stronger convergence results for certain 
explicit choices of weights.  In the language of statistical mechanics
these results give an explicit description of the finite size effects.

\begin {theorem}\label{th2}
For the branching weights $w_2 = \lambda$ and $w_{n} = (n-1)!,~n\neq 2$, the
partition function satisfies
\begin {equation}\label{th2z}
 \frac{Z_N}{e^\lambda (N-1)!} \rightarrow 1
\end {equation}
and
\begin {equation} \label{poisson}
 N-\sigma(s)\xrightarrow[]{~d~} \Pois(\lambda)
\end {equation}
as $N\rightarrow \infty$.
Moreover, 
the tree $T_N$ consists of $r$,
$s$, and $\gs(s)-1$ branches attached to $s$;
with probability tending to $1$,  $N-\gs(s)$ of these branches have
size $2$ and all other have size $1$ (i.e., they contain a single leaf
only).  
\end {theorem}

Note that in the limit $N\to\infty$, the branches of size 2 disappear to
infinity, so we do not see them in the limit given by Theorem \ref{th1}.

\begin {theorem} \label{th3}
Let the branching weights be $w_{n} = \left((n-1)!\right)^\alpha$, where
$0<\alpha < 1$.
Then the partition function satisfies 
\begin {equation}\label{th3z}
 {Z_N}=\bigpar{(N-1)!}^\alpha\exp\bigpar{O(N^{1-\ga})}
=\exp\bigpar{\ga N\log(N)-\ga N+O(N^{1-\ga})}.
\end {equation}
Furthermore, with probability tending to $1$,
the random tree $T_N$ has the following properties, with $K=\floor{1/\ga}$:
\begin{romenumerate}
  \item\label{th3s}
$\gs(s)=N-O(N^{1-\ga})$.
  \item\label{th3K}
All vertices except $s$ have degrees $\le K+1$.
  \item\label{th3trees}
All subtrees attached to $s$ have sizes $\le K+1$.
\end{romenumerate}
Moreover,
let\/ $X_{i,N}$ be the number of vertices of degree $i$ in $T_N$
and let 
\begin{equation}\label{ni}
n_i=i!^\ga N^{1-i\ga}.
\end{equation}
\begin{romenumerateq}
  \item\label{th3m}
If\/ $1\le i<1/\ga$, then $n_i\to\infty$ as $N\to\infty$ and
\begin{equation}\label{th3pto}
\frac{ X_{i+1,N}}{n_i}\pto 1.  
\end{equation}
If\/ $i=1/\ga=K$ (which occurs only when $1/\ga$ is an integer), then
$n_K=K!^\ga$ is constant and
\begin{equation}\label{th3pois}
X_{K+1,N}\dto \Pois(n_K).
\end{equation}
\end{romenumerateq}
\end {theorem}

With these branching weights,
the asymptotic distributions of the numbers $X_{i,N}$ of vertices of
different degrees are Gaussian, except in the Poisson case when
\eqref{th3pois} applies.

\begin{theorem}
  \label{th4}
Let
$w_{n} = \left((n-1)!\right)^\alpha$ with $0<\ga<1$ as in Theorem \ref{th3}. Then
there exist numbers $\nx_i=\nx_i(N)=\etto n_i$,
$1\le i<1/\ga$,
with $n_i$ given by \eqref{ni}, such that, as
\Ntoo,
\begin{align}\label{th4gauss}
  \frac{X_{i+1,N}-\nx_i}{\sqrt{n_i}}&\dto \cN(0,1), && 1\le i<1/\ga,
\\
X_{i+1,N}&\dto \Pois(n_i), && i=K=1/\ga.
\end{align}
Moreover, these hold jointly for all $i\le K$, with independent limits.

More precisely, for each $i<1/\ga$,
\begin{equation}\label{thx1}
  \nx_i=n_i\bigpar{1-(1-i\ga)N^{-\ga}+O\bigpar{N^{-2\ga}}}+O(1).
\end{equation}

In particular,  when $\ga$ is not too small, we have the explicit limits
  \begin{align}\label{thx2a}
 \frac{X_{2,N}-N^{1-\ga}}{N^{(1-\ga)/2}}&\dto \cN(0,1), 
&& 1>\ga>\tfrac13,
\\
\frac{X_{2,N}-(N^{1-\ga}-(1-\ga)N^{1-2\ga})}{N^{(1-\ga)/2}}&\dto \cN(0,1), 
&& 1>\ga>\tfrac15,	
\\
 \frac{X_{3,N}-2^\ga N^{1-2\ga}}{N^{(1-2\ga)/2}}&\dto \cN(0,2^\ga), 
&& \tfrac12>\ga>\tfrac14,
\\
 \frac{X_{3,N}-(2^\ga N^{1-2\ga}-(1-2\ga)2^\ga N^{1-3\ga})}{N^{(1-2\ga)/2}}&\dto \cN(0,2^\ga), 
&& \tfrac12>\ga>\tfrac16,	
\\
 \frac{X_{4,N}-6^\ga N^{1-3\ga}}{N^{(1-3\ga)/2}}&\dto \cN(0,6^\ga), 
&& \tfrac13>\ga>\tfrac15,  
\\
\frac{X_{4,N} - (6^\alpha N^{1-3\alpha} - (1-3\alpha) 6^\alpha N^{1-4\alpha})}
  {N^{(1-3\alpha)/{2}}}&\dto \cN(0,6^\ga),  
&&\tfrac13>\ga>\tfrac17.  \label{thx4b}
  \end{align}
\end{theorem}

For smaller $\ga$, it is possible to obtain further terms in the expansion
of $\nx_i$, and thus explicit forms of the asymptotic mean of $X_{i+1,N}$.
However, this
approach seems to become more and more difficult as $\ga$ becomes smaller.

\begin{remark}\label{Rth4}
  The proof of \eqref{th4gauss}
shows also the stronger result that
the joint distribution of $(X_{i+1,N})_{i=1}^K$ can be approximated by 
the joint distribution of independent Poisson random variables
$Y_{i,N}\sim\Pois(\nx_i)$, in the sense that the 
total variation distance tends to 0 as $N\to\infty$:
\begin{equation}
\frac12 \sum_{m_1,\dots,m_K}\Bigabs{\P(X_{i+1,N}=m_i, \,\forall i) 
- \P(Y_{i,N}=m_i, \,\forall i) }
\to0.
\end{equation}
\end{remark}

\begin{remark}\label{Rz}
  The estimate \eqref{th3z} of the partition function can be improved to
\begin{equation}\label{rz}
 Z_N={(N-1)!^\ga}
 \exp\lrpar{N^{1-\ga}+\Bigpar{2^\ga-\frac{1-\alpha}2} N^{1-2\ga}
+ O\bigpar{N^{1-3\ga}}
+o(1)}.
\end{equation}
In particular, if $1>\ga>\frac12$, then 
$
Z_N={(N-1)!^\ga} \exp\bigpar{N^{1-\ga}+o(1)}$.

Again it seems possible, but more complicated, to obtain further terms in the
exponent. 
\end{remark}

\begin {remark} \label{Rag1}
 It is straightforward to show, using the same methods as in the proof of 
 Theorem \ref{th2}, that when $\alpha > 1$
\begin {equation}
 Z_N = (N-1)!^\alpha (1+o(1))
\end {equation}
and all the branches which are attached to $s$ have size 1, with a 
probability which tends to 1 as $N\rightarrow\infty$. In this case the 
leading contribution to the partition function comes only from the
Boltzmann factor 
of the vertex $s$, i.e.~$w_{\sigma(s)}$. The case $\alpha =1$ is a 
marginal case when larger branches start to appear and their entropy 
adds a contribution to the partition function which appears in the 
associated exponential.
\end {remark}

\section{Proofs of theorems}\label{Spf}
In this section we state and prove a few lemmas and prove 
Theorems \ref{th1}--\ref{th3}. In the following we will always assume that the branching weights satisfy the condition in Equation (\ref{condition1}). Define
\begin {equation}
 Z(N,n) = \sum_{d_1+\cdots+d_N = n} \prod_{i=1}^N w_{d_i+1}.
\end {equation}
By Lagrange's inversion formula \cite{sdf,flajolet} 
(or by a combinatorial argument, see  \cite{Dwass,Kolchin,Pitman:enum}),
it holds that
\begin {equation} \label{partitionfunction}
 Z_N = \frac{1}{N} Z(N,N-1).
\end {equation}
More generally the partition function for an ordered forest of $m$ trees with a total number of edges $N$ is
\begin {equation}\label{dwassm}
 Z_N^{(m)} = \frac{m}{N} Z(N,N-m).
\end {equation}

\begin {lemma} \label{l1}
 For every $\epsilon >0$ there exists a $C_\epsilon < \infty$ such that for all $N$ and $n$
\begin {equation}
 Z(N,n) \leq \epsilon Z(N,n+1) + C_\epsilon ^N.
\end {equation}
\end {lemma}
\begin {proof}
Consider a finite sequence $d_1,\ldots,d_N$ for which $\sum_{i} d_i = n$. Let $i^\ast$ be the smallest index such that $d_{i^{\ast}} = \max_i{d_i}$. Define a sequence
\begin {equation}
d_i^{\ast} = \left\{\begin {array}{ll} 
d_i + 1 & \text{if $i = i^\ast$}, \\
d_i & \text{otherwise.}
\end {array}\right.
\end {equation}
 Note that $d_{i^{\ast}}^\ast$ is the unique maximum in $(d_i^{\ast})$, so
 $(d_i)$ can be recovered from $(d_i^{\ast})$ 
and the map $(d_i)\mapsto(d_i^*)$ is injective.

Let $\epsilon > 0$ be given. Choose a number $A_\epsilon$ such that $w_i/w_{i+1} < \epsilon$ if \\ $i \geq A_\epsilon$. Then
\begin {equation}
 \sum_{\substack{d_1 + \cdots + d_N = n \\ \max_i d_i > A_\epsilon }} \prod_{i=1}^{N} w_{d_i+1} \leq \epsilon \sum_{\substack{d_1^{\ast} + \cdots + d_N^{\ast} = n+1}} \prod_{i=1}^{N} w_{d_i^{\ast}+1} \leq \epsilon Z(N,n+1)
\end {equation}
and, crudely,
\begin {equation}
\sum_{\substack{d_1 + \cdots + d_N = n \\ \max_i d_i \leq A_\epsilon }} \prod_{i=1}^{N} w_{d_i+1} \leq \left(\sum_{i=0}^{A_\epsilon} w_{i+1}\right)^N.
\end {equation}
Taking $C_\epsilon = \sum_{i=0}^{A_\epsilon} w_{i+1}$ completes the proof.
\end {proof}

\begin {lemma} \label{l:sdiverges}
As  $N\rightarrow \infty$, $\sigma(s) \xrightarrow[]{p} \infty$.
\end {lemma}
\begin {proof}
It suffices to show that 
\begin {equation} \label{conv1}
\nu_N(\sigma(s) = k) \rightarrow 0 
\end {equation}
for every fixed $k\geq 1$, since $\nu_N(\sigma(s) \geq m) = 1-
\sum_{k=1}^{m-1}\nu_N(\sigma(s) = k)$. If the vertex $s$, in a tree with $N$
edges, has degree $k+1$, then removing $s$ and $r$ but
leaving all edges from $s$ to its children as pendant edges, cf.\ Remark
\ref{Rroot},
leaves a forest with $k$
trees and $N-1$ edges. Therefore, using \eqref{partitionfunction} and 
\eqref{dwassm}, 
\begin {equation} \label{rsr}
\nu_N(\sigma(s) = k+1) = \frac{N}{N-1} k w_{k+1} \frac{Z(N-1,N-k-1)}{Z(N,N-1)}.
\end {equation}
Let $\epsilon > 0$ be given. Use Lemma \ref{l1} $k$ times to get
\begin {equation}
Z(N-1,N-k-1) \leq \epsilon^k Z(N-1,N-1) + k C_\epsilon^N
\end {equation}
and note that 
\begin {equation}
Z(N,N-1) \geq w_1 Z(N-1,N-1).
\end {equation}
 Since the branching weights satisfy (\ref{condition1}), 
$Z(N-1,N-1)\ge w_Nw_1^{N-2}$ grows
 super exponentially and in particular $Z(N-1,N-1) \geq (2C_\epsilon)^{N}$
 for $N$ large enough. 
Therefore
\begin {equation}
 \nu_N(\sigma(s) = k+1) \leq \frac{N}{N-1} k w_{k+1} (w_1^{-1} \epsilon^k + k 2^{-N})
\end {equation}
and (\ref{conv1}) follows since $\epsilon$ is arbitrary.
\end {proof}
\begin {lemma} \label{lsum}
 For any $N\geq 1$ and $n \geq 0$
\begin {equation}
 \sum_{\ell = 0}^N \ell w_{\ell+1}Z(N-1,n-\ell) = \frac{n}{N} Z(N,n).
\end {equation}
\end {lemma}
\begin {proof}
 \begin {eqnarray}\nonumber
  && \sum_{\ell = 0}^n \ell w_{\ell+1}Z(N-1,n-\ell) = \sum_{d_1+\cdots + d_{N-1} + \ell = n} \ell w_{\ell+1} \prod_{i=1}^{N-1} w_{d_i+1} \\
&& = \sum_{d_1+\cdots + d_{N} = n} d_N  \prod_{i=1}^{N} w_{d_i+1}.
 \end {eqnarray}
By symmetry we can replace $d_N$ in front of the product by any $d_j$, $j=1,\ldots,N$. Summing over $j$ then gives the desired result.
\end {proof}

\begin {lemma} \label{l:2overL}
Assume $N>1$ and let $s_1$ be the first child of $s$. 
 If $L\geq 1$ and $k\geq L$, then
\begin {equation}
 \nu_N\bigpar{L+1 \leq \sigma(s_1) \leq k+1 ~|~ \sigma(s)  = k+1} 
\leq \frac{2}{L}.
\end {equation}
\end {lemma}
\begin {proof}
If $\sigma(s) = k + 1 \geq 2$ and $\sigma(s_1) = \ell + 1 \geq 1$ then
removing the vertices $r$, $s$ and $s_1$, again leaving pendant edges, 
leaves a forest with $k+\ell -1$
trees and $N-2$ edges. Therefore (assuming $N\geq 3$),  
\begin {multline} \label{rssr}
 \nu_N(\sigma(s) = k+1, \sigma(s_1) = \ell+1) 
\\= \frac{N (k+\ell-1) w_{k+1}w_{\ell+1}}{N-2}  \frac{Z(N-2,N-1-k-\ell)}{Z(N,N-1)}.
\end {multline}
 By (\ref{rsr}) and (\ref{rssr}),
\begin {multline} \label{prob2} 
 \nu_N(\sigma(s_1) = \ell+1 ~|~ \sigma(s)  = k+1) 
\\= \frac{(N-1)(k+\ell-1) w_{\ell+1} }{(N-2)k} \frac{Z(N-2, N-1-k-\ell)}{Z(N-1,N-1-k)}.
\end {multline}
By Lemma \ref{lsum},
\begin{equation}
  \begin{split}
  \sum_{\ell \geq L} w_{\ell+1} Z(N-2,N-1-k-\ell) 
&\leq \frac{1}{L} \sum_{\ell \geq 0} \ell w_{\ell+1} Z(N-2,N-1-k-\ell) \\
 &= \frac{1}{L} \frac{N-1-k}{N-1} Z(N-1,N-1-k).	
  \end{split}
\end{equation}
Hence, (\ref{prob2}) implies
\begin {equation}
 \sum_{\ell = L}^{k}  \nu_N(\sigma(s_1) = \ell+1 ~|~ \sigma(s)  = k+1) 
\leq \frac{N-1-k}{N-2} \frac{2}{L} \leq \frac{2}{L}.
\end {equation}
\end {proof}
\begin {lemma} \label{l:s1is1}
As $N\rightarrow \infty$, $\nu_N(\sigma(s_1)=1) \rightarrow 1$.
\end {lemma}
\begin{proof}
 Fix $L>1$ and an $\ell$ such that $1 \leq \ell < L$. Note that when $\ell \geq 1$ the formula (\ref{rssr}) is symmetric in $k$ and $\ell$. Therefore
 \begin{equation}
   \begin{split}
 \nu_N(\sigma(s_1) = \ell+1) &= \sum_{k=1}^\infty \nu_N(\sigma(s)=k+1, \sigma(s_1) = \ell+1) \\
&= \nu_N(\sigma(s) = \ell+1, \sigma(s_1) \geq 2) \leq \nu_N(\sigma(s) = \ell+1)
	    \end{split}
 \end{equation}
and thus $\nu_N(\sigma(s_1)=\ell+1) \rightarrow 0$ as $N\rightarrow \infty$ by Lemma \ref{l:sdiverges}. Next, Lemma \ref{l:2overL} implies
\begin{equation}
  \begin{split}
 \nu_N&(L+1 \leq \sigma(s_1) \leq \sigma(s)) \\ 
& = \sum_k \nu_N(L+1 \leq \sigma(s_1) \leq k+1~|~\sigma(s)=k+1) \nu_N(\sigma(s) = k+1) \leq \frac{2}{L}.
  \end{split}
\end{equation}
Thus
\begin{equation}
  \begin{split}
 \limsup_{N\rightarrow\infty} \nu_N(2 \leq \sigma(s_1) \leq \sigma(s)) 
&\leq \limsup_{N\rightarrow\infty} \left(\sum_{\ell = 1}^{L-1} \nu_N(\sigma(s_1) = \ell+1) + \frac{2}{L}\right) \\
&= \frac {2}{L}.
  \end{split}
\raisetag\baselineskip
\end{equation}
Since $L$ is arbitrary, $\nu_N(2 \leq \sigma(s_1) \leq \sigma(s)) \rightarrow 0$ as $N\rightarrow \infty$. By the symmetry of (\ref{rssr}) in $k$ and $\ell$ we also find that
\begin {equation}
\nu_N(2 \leq \sigma(s) \leq \sigma(s_1))  = \nu_N(2 \leq \sigma(s_1) \leq \sigma(s)) \rightarrow 0
\end {equation}
as $N\rightarrow \infty$. Finally, since $\sigma(s) \geq 2$, we have
\begin {equation}
 \nu_N(\sigma(s_1)\geq 2) \leq \nu_N(2 \leq \sigma(s) \leq \sigma(s_1))  + \nu_N(2 \leq \sigma(s_1) \leq \sigma(s)) \rightarrow 0
\end {equation}
as $N\rightarrow \infty$.
\end{proof}
\begin {proof}[Proof of Theorem \ref{th1}]
Let $R>0$.
By Lemma \ref{l:sdiverges}, $\gs(s)\pto\infty$, so $\nu_N\bigpar{\gs(s)\ge
R}\to1$. 
Given that $\sigma(s) \geq R$, denote the
  first $R-1$ children of $s$ by $s_1,\ldots,s_{R-1}$. Then by 
Lemma \ref{l:s1is1} and symmetry 
$\nu_N\bigpar{\gs(s_i)=1,\,\gs(s)\ge R}\to1$ for every $i\le R$ and thus
we find that
\begin {equation} \label{XXX}
 \nu_N(\sigma(s) \geq R,\,\sigma(s_1)=\cdots=\sigma(s_{R-1}) = 1) \rightarrow 1
\end {equation}
as $N\rightarrow \infty$.
Since $R$ is arbitrary, the result follows from the definition of the
topology on $R$, cf.\ the comment below (\ref{defmetric}).
\end {proof}
\begin {proof}[Proof of Theorem \ref{th2}]
First, we establish an upper bound on $Z_N$. Consider Equation
(\ref{partitionfunction}) for $Z_N$. For a given sequence $(d_i)$, let $m_j$
denote the number of indices $i$ for which $d_i = j$ where
$j=0,\ldots,N-1$. Instead of summing over $(d_i)$ we sum over $(m_j)$. For a
given sequence $(m_j)$ there are $\binom{N}{m_0,\ldots,m_{N-1}}$ sequences
$(d_i)$ and therefore, since $w_1=1$, 
\begin {eqnarray} \nonumber
 \frac{Z_N}{(N-1)!}&=& \sum_{\substack{m_0 + \cdots + m_{N-1} = N \\m_1 + 2m_2 + \cdots + (N-1) m_{N-1} = N-1 }} \prod_{i=0}^{N-1} \frac{w_{i+1}^{m_i}}{m_i!}\\
&=& \sum_{\substack{m_1 + 2m_2 + \cdots + (N-1) m_{N-1} = N-1 }} \frac{1}{\left(N-\sum_{j=1}^{N-1}m_j\right)!}\prod_{i=1}^{N-1} \frac{w_{i+1}^{m_i}}{m_i!}.\nonumber \\ \label{ZNfactorial}
\end {eqnarray}
Denote the maximum vertex degree by $M$ and fix a number $K \geq 2$. By
Lemma \ref{l:sdiverges}, it is sufficient to consider the case $M>K$. That
contribution to (\ref{ZNfactorial}) can be estimated 
by shifting $m_{M-1}\to m_{M-1}+1$ which yields the upper bound
\begin{equation}
  \begin{split}
& \sum_{\substack{m_1 + 2m_2 + \cdots + (M-1) m_{M-1} = N-M \\ M>K }} \frac{(M-1)!}{\left(M+\sum_{j=1}^{M-1}(j-1)m_j-1\right)!}\prod_{i=1}^{M-1} \frac{w_{i+1}^{m_i}}{m_i!} \\
&\qquad\qquad \leq \sum_{\substack{m_1 + 2m_2 + \cdots + (M-1) m_{M-1} = N-M \\ M>K }} \frac{\lambda^{m_1}}{m_1!}\prod_{i=2}^{M-1} \frac{(i!/M^{i-1})^{m_i}}{m_i!} \label{upperboundII} \\ 
&\qquad\qquad \leq \exp\left(\lambda + \sum_{i=2}^\infty\frac{i!}{(i\vee K)^{i-1}}\right)
 \end{split}
\raisetag{1.5\baselineskip}
\end{equation}
where $A\vee B$ denotes the maximum of $A$ and $B$. The last expression converges to $e^\lambda$ when $K\rightarrow\infty$ by dominated convergence.

Next we establish a corresponding
lower bound on $Z_N$. Consider the contribution to
(\ref{ZNfactorial}) from terms for which the only nonzero elements in the
sequence $(m_i)$ are $m_0$, $m_1$ and $m_k=1$ where $k\geq 2$ is
arbitrary;
thus $m_0=k$, $m_1=N-k-1$ and $m_k=1$.
 These terms provide the following lower bound of
(\ref{ZNfactorial}) 
\begin {equation} \label{lowerbound}
\sum_{k=2}^{N-1}w_{k+1}\frac{1}{k!} \frac{w_2^{N-1-k}}{(N-1-k)!} =  \sum_{\ell=0}^{N-3}\frac{\lambda^{\ell}}{\ell!} \rightarrow e^\lambda
\end {equation}
as $N\rightarrow \infty$.
This and \eqref{upperboundII} prove \eqref{th2z}.

To complete the proof, note that the probability that $T_N$ has $\gs(s)=N-j$
and that exactly $j$ of the $\gs(s)-1=N-j-1$ branches attached to $s$ have
size 2 and all others size 1 is, assuming $N>2j$ and using \eqref{th2z},
\begin{equation}
  \frac1{Z_N}\binom{N-j-1}{j}w_1^{N-j-1}w_2^{j}w_{N-j}
=
  \frac1{Z_N}\binom{N-j-1}{j}\gl^{j}(N-j-1)!
\to \frac{\gl^j}{j!}e^{-\gl}.
\end{equation}
These limits sum to 1 and yield the $\Pois(\gl)$ distribution in
\eqref{poisson}. 
\end {proof}

\begin{proof}[Proof of Theorem \ref{th3}]
Consider the weights $w_{n+1} = n!^\alpha$.
Write, again by \eqref{partitionfunction},
\begin {equation}
Z_N = \frac{1}{N} \sum_{d_1+\cdots+d_N = N-1} \prod_{i=1}^N d_i!^\ga.
\end {equation}
We get the lower bound 
\begin {equation} \label{an}
Z_N \geq (N-1)!^\alpha 
\end {equation}
by considering only the terms in $Z_N$ with one $d_i = N-1$, and all others
0
(i.e., stars).

Define $Z_N(k,\epsilon)$ as the contribution to $Z_N$ when precisely $k\geq 0$ vertices have degree greater than $\epsilon (N-1)$ where $\epsilon$ is some small positive number. First consider the case when $k=0$. Let $(d_i)_{i=1}^{N}$ be a sequence  which satisfies $d_i \leq \epsilon (N-1)$ for all $i$. Using the simple relation
\begin {equation} \label{fact}
 N!\,M! \leq (N+1)!\,(M-1)!, \qquad \qquad \text{for}~N\geq M-1
\end {equation}
we can distribute and add the smallest elements in $(d_i)_{i=1}^N$ to the
larger ones until each of them reaches $\epsilon (N-1)$. Thus we obtain the
upper bound,
using Stirling's formula, 
\begin {equation}
 \prod_{i=1}^N d_i!^\alpha 
\leq \ceil{\epsilon (N-1)}!^{\alpha/\epsilon} 
\leq \CC N^{2\ga/\eps} (N-1)!^\alpha \epsilon^{\alpha N}
\end {equation}
where $\CCx>0$ is a number independent of $N$ (but, as other constants
below, it may depend on $\ga$ and $\eps$). 
Therefore,
\begin {equation}\label{k=0}
 Z_N(0,\epsilon) 
\leq \CCx N^{2\ga/\eps} (N-1)!^\alpha \epsilon^{\alpha N} \binom{2N-2}{N-1} 
\leq \CCx N^{2\ga/\eps} (N-1)!^\alpha \epsilon^{\alpha N} 2^{2N}
\end {equation}
 which is negligible compared to (\ref{an}) as $\Ntoo$ for $\epsilon$ small enough.

 Next consider the case when two or more of the $d_i$ are larger than
 $\epsilon (N-1)$, i.e.~ when $k\geq 2$ in $Z_N(k,\epsilon)$. Clearly,
 $k<1/\epsilon$. Denote the $d_i$ which are greater than $\epsilon(N-1)$ by
 $d_{i_1},\ldots,d_{i_k}$ and let $D_j=d_{i_j}$. 
The indices $i_j$  can be chosen in $\binom{N}{k}$ ways.   
We will now lump together all the $D_i$ into a single one, i.e.~ we define
$$D = D_1 + \cdots + D_k.$$
For each $D$, there are at most $\binom{D+k-1}{k-1}$ choices of
$D_1,\dots,D_k$. 
Note that, with $\dx=\ceil{\eps(N-1)}$, using $D_i\ge\dx$ and Stirling's
formula again, 
\begin {eqnarray}
 \frac{D_1!\cdots D_k!}{D!} 
\le \frac{\dx^k}{(k\dx)!} 
\leq \CC N^{k} \left(\frac{1}{k}\right)^{k\epsilon N}
\end {eqnarray}
where $\CCx>0$ is independent of $N$. Thus, we get the upper bound
\begin{equation}\label{kge2}
  \begin{split}
\sum_{2 \leq k \leq 1/\epsilon} Z_N(k,\epsilon) 
&\leq  \CCx^\alpha \sum_{2 \leq k \leq 1/\epsilon} \binom{N}{k}N^{\alpha k} \left(\frac{1}{k}\right)^{\alpha k\epsilon N}  \\
  &\qquad\times\!\!\! ~\sum_{\substack{D+d_1+\cdots+d_{N-k} = N-1 \\ D > \epsilon (N-1), ~d_i \leq \epsilon (N-1), ~\forall i}} \binom{D+k-1}{k-1} D!^\alpha\prod_{i=1}^{N-k} d_i!^\alpha  \\
 &\leq \CC N^{3/\epsilon} \left(\frac{1}{2}\right)^{2\alpha \epsilon N} Z_N(1,\epsilon).	
  \end{split}
\raisetag\baselineskip
\end{equation}
where $\CCx>0$ is independent of $N$. This estimate, together with
\eqref{k=0},
shows that the main contribution to $Z_N$ for $N$ large comes from $Z_N(1,\epsilon)$. 

Finally, we consider $Z_N(1,\epsilon)$.  Using the representation as in
(\ref{ZNfactorial}) we have, writing $\neps=\floor{\eps(N-1)}$ for convenience,
\begin {equation}\label{newton}
\frac{ Z_N(1,\epsilon)}{(N-1)!}
= \sum_{D=\neps+1}^{N-1} 
\sum_{
m_1 + 2 m_2 + \cdots + \neps m_{\neps}=N-1-D}
  \frac{D!^\alpha}{(N-1-\sum_{j=1}^\neps m_j)!} 
  \prod_{i=1}^{\neps} \frac{i!^{\alpha m_i}}{m_i!}
\end {equation}
where $D+1$ denotes the degree of the large vertex and  
$m_i$ denotes the number of vertices which have degree $i+1$. 
Consider one term in this sum and let $\tD=D+\sum_{i=K+1}^\neps{im_i}$,
adding the outdegrees of all vertices which have degree greater than $K+1$
to the large vertex.
Then
\begin{equation}
  \tD!\ge D!\cdot D^{\sum_{i=K+1}^\neps{im_i}}
\ge D!\cdot \neps^{\sum_{i=K+1}^\neps{im_i}}
\end{equation}
and
\begin{equation}
\Bigpar{N-1-\sum_{i=1}^K{m_i}}!
\le 
\Bigpar{N-1-\sum_{i=1}^\neps{m_i}}!\cdot N^{\sum_{i=K+1}^\neps{m_i}}
\end{equation}
Thus
\begin{equation}
\frac{D!^\ga}{(N-1-\sum_{i=1}^\neps{m_i})!}
\le 
\frac{\tD!^\ga\,N^{\sum_{i=K+1}^\neps{m_i}}}
 {(N-1-\sum_{i=1}^K{m_i})!\,\neps^{\ga\sum_{i=K+1}^\neps{im_i}}}
\end{equation}
and
\begin {multline}\label{ini}
\frac{ Z_N(1,\epsilon)}{(N-1)!}
\le \sum_{\tD=\neps+1}^{N-1} 
\;
\sum_{m_1  + \cdots + K m_{K}=N-1-\tD}
\;
\sum_{m_{K+1},\dots,m_{\neps}\ge0}
\\
  \frac{\tD!^\alpha}{(N-1-\sum_{j=1}^K m_j)!} 
  \prod_{i=1}^{K} \frac{i!^{\alpha m_i}}{m_i!}
  \prod_{i=K+1}^{\neps} \parfrac{Ni!^\alpha}{\neps^{i\ga}}^{m_i}\frac1{m_i!}
.
\end {multline}
We have
\begin{equation}
  \begin{split}
\sum_{m_{K+1},\dots,m_{\neps}}\;
  \prod_{i=K+1}^{\neps} \parfrac{Ni!^\alpha}{\neps^{i\ga}}^{m_i}\frac1{m_i!}
&=
  \prod_{i=K+1}^{\neps}
\exp\parfrac{Ni!^\alpha}{\neps^{i\ga}}	
\\&=
\exp\lrpar{
\sum_{i=K+1}^{\neps}
\frac{Ni!^\alpha}{\neps^{i\ga}}	}.
  \end{split}
\end{equation}
Let, using $\neps=\floor{\eps(N-1)}>\eps N/2$ (assuming $N$ large), 
\begin{equation}
  a_i=\frac{Ni!^\ga}{\neps^{i\ga}}
\le
\frac{2^{i\ga}i!^\ga}{\eps^{i\ga}}N^{1-i\ga}.
\end{equation}
Noting that $a_{i+1}/a_i=((i+1)/\neps)^\ga\le1$ for $i<\neps$, 
we find
\begin{equation}
  \begin{split}
  \sum_{i=K+1}^\neps a_i \le (K+1)a_{K+1}+Na_{2K+2}
&=O\lrpar{N^{1-(K+1)\ga}}+O\lrpar{N^{2-(2K+2)\ga}} 
\\&=o(1)	
  \end{split}
\end{equation}
and thus from \eqref{ini},
\begin {multline}
\frac{ Z_N(1,\epsilon)}{(N-1)!}
\le 
\bigpar{1+o(1)}\sum_{\tD=\neps+1}^{N-1} 
\;
\sum_{m_1  + \cdots + K m_{K}=N-1-\tD}
\\
  \frac{\tD!^\alpha}{(N-1-\sum_{j=1}^K m_j)!} 
  \prod_{i=1}^{K} \frac{i!^{\alpha m_i}}{m_i!}
.
\end {multline}
The sum here is just the sum in \eqref{newton} with $m_i=0$ for $i>K$, so we
have shown that $Z_N(1,\eps)$ is dominated by such terms. Recalling
\eqref{k=0} and \eqref{kge2} we see that 
\begin {equation}\label{isaac}
\frac{ Z_N}{(N-1)!}
= 
\bigpar{1+o(1)}\!\!\!
\sum_{m_1  + \cdots + K m_{K}< N-\neps-1}
  \frac{(N-1-\sum_{j=1}^K jm_j)!^\ga} {(N-1-\sum_{j=1}^K m_j)!} 
  \prod_{i=1}^{K} \frac{i!^{\alpha m_i}}{m_i!}
\end {equation}
and that
$Z_N$ is dominated by trees having
exactly one vertex of degree $>\eps(N-1)$ and all other vertices having
degrees $\le K+1$.

By Lemma \ref{l:sdiverges}, the contribution from trees with $\gs(s)\le K+1$ is
negligible, so it suffices to consider the case when the unique vertex with
high degree is $s$, which proves \ref{th3K}.

To obtain the more precise results in \ref{th3s} and \ref{th3m}, fix $i\le
K$, fix $m_j$ for $j\neq i$, and denote the summand in \eqref{isaac} by
$b(m_i)$.
Increasing $m_i$ by 1 decreases $D=N-1-\sum_{j=1}^K jm_j$ by $i$ and,
assuming still $D>\neps$ and recalling the definition of $n_i$ in \eqref{ni},
\begin{equation}
  \label{winston}
\frac{b(m_i+1)}{b(m_i)}
\le N \neps^{-i\ga} \frac{i!^{\alpha}}{m_i+1}
\le \CC  \frac{N^{1-i\ga} i!^{\alpha}}{m_i+1}
= \CCx  \frac{n_i}{m_i+1}\ .
\CCdef\CCwinston
\end{equation}
If $m_i\ge\floor{2\CCx n_i}$, this ratio is less than 1/2. In particular,
\begin{equation}
  \sum_{m_i\ge 3\CCx n_i} b(m_i) \le 2 b(\floor{3\CCx n_i})
\le 2^{2-\CCx n_i} b(\floor{2\CCx n_i}).
\end{equation}

If $i<1/\ga$, then $n_i\to\infty$ as \Ntoo.
Summing over all $m_j$, $j\neq i$, we see that the contribution to $Z_N$
from $m_i\ge 3\CCx n_i$ is negligible, so we may assume that $m_i<3\CCx n_i$.
In the exceptional case $i=1/\ga$, we obtain by the same argument that we
may assume $m_i<\log N$, say. In particular, since 
$n_i=O(N^{1-i\ga})=O(N^{1-\ga})$, we see that we may assume
$\gs(s)=D+1=N-\sum_{j=1}^K jm_j=N-O(N^{1-\ga})$, which proves \ref{th3s}.

For the remaining terms, we now may use 
$D=N-o(N)$ to 
improve \eqref{winston} to 
\begin{equation}
  \label{chu}
\frac{b(m_i+1)}{b(m_i)}
=(1+o(1)) N\cdot N^{-i\ga} \frac{i!^{\alpha}}{m_i+1}
= (1+o(1))  \frac{n_i}{m_i+1}.
\end{equation}
Assume $i<1/\ga$ and let $\gd>0$. 
We can repeat the argument above, using \eqref{chu} instead of
\eqref{winston} and $(1+\gd/2)n_i$ instead of $2\CCx n_i$, and conclude that
the 
terms with $m_i\ge(1+\gd)n_i$ are negligible. Similarly, \eqref{chu}
implies also that the terms with $m_i\le(1-\gd)n_i$ are negligible. 
Hence, $Z_N$ is dominated by terms with $(1-\gd)n_i<m_i<(1+\gd)n_i$.
Since $X_{i+1,N}=m_i$, this proves \ref{th3m} for $i<1/\ga$.

If $1/\ga$ is an integer and $i=K=1/\ga$, then it follows from \eqref{chu}
in the same way that $m_K$ is stochastically bounded and that 
$\nu_N\set{m_K=m+1}/\nu_N\set{m_K=m}\to n_K/(m+1)$ for every $m$, which
implies that $m_K\dto\Pois(n_K)$, completing the proof of \ref{th3m}.

Furthermore, \eqref{winston} implies, for all $m_i$ such that $D>\neps$,
\begin{equation}
  \frac{b(m_i)}{b(0)}
\le 
\frac{(\CCwinston n_i)^{m_i}}{m_i!}.
\end{equation}
Using this for each $i\le K$, we see that the general summand  in
\eqref{isaac} is at most
$\prod_{i=1}^K\frac{(\CCwinston n_i)^{m_i}}{m_i!}$,
and thus \eqref{isaac} yields
\begin{equation}
\frac{Z_N}{(N-1)!}
\le(1+o(1)) \frac{(N-1)^\ga}{(N-1)!}
\sum_{m_1,\dots,m_K}  \prod_{i=1}^K\frac{(\CCwinston n_i)^{m_i}}{m_i!}
\end{equation}
and
\begin{equation}
\frac{Z_N}{(N-1)!^\ga}
\le(1+o(1)) 
  \prod_{i=1}^K\exp(\CCwinston n_i)
=  \exp\Bigpar{\sum_{i=1}^K\CCwinston n_i+o(1)},
\end{equation}
which proves \eqref{th3z}.

Finally, we show \ref{th3trees}.
If $\tau$ is a tree in $\gG_N$ such that all vertices except $s$ have
degrees $\le K+1$, but some branch attached to $s$ has more than $K+1$
vertices, pick the first such branch and find, by breadth-first search, say,
a subtree $\tau_0$ of that branch with exactly $K+2$ vertices.
Rearrange the edges inside $\tau_0$ so that $\tau_0$ is replaced by a star
with center adjacent to $s$; this produces a vertex of degree $K+2$. 
Let $\tau'\in\gG_N$ be the result of making this change inside $\tau$.
We
have changed the degree of (at most) $K+2$ vertices, and since all old and
new degrees are at most $2K+1$, the weights of $\tau$ and $\tau'$ differ by
at most a constant factor. Furthermore, $\tau'$ has exactly one vertex of
degree $K+2$, and thus only a bounded number of trees $\tau$ can produce the
same $\tau'$. Consequently, 
\begin{multline}
  \P(T_n \text{ has a branch of size }>K+1)
\\\le
\CC
  \P(T_n \text{ has a vertex $\neq s$ of degree }>K+1),
\end{multline}
and this probability tends to 0 by \ref{th3K}.
\end{proof}

\begin{proof}[Proof of Theorem \ref{th4}]
Recall that $Z_N$ is given by \eqref{isaac}, and that the significant terms
have $m_i=\etto n_i =O(N^{1-i\ga})$, except when $i=1/\ga$.

Let us first note that if $1/\ga$ is an integer and $i=K=1/\ga$, then, see
the proof of Theorem \ref{th3}, \eqref{chu} implies that $X_{K+1,N}=m_K$ has
an asymptotic Poisson distribution $\Pois(n_K)$, which further is
asymptotically independent of $X_{i,N}$, $i\le K$; furthermore,
$\sum_{m_K}b(m_K)=\exp(n_K+o(1))b(0)$.
In the sequel we thus assume $m_K=0$ and sum only over $m_i$, $i<K$,
when $i=K=1/\ga$; we omit the trivial modifications below in this case.

Define, for a fixed $\eta\in(0,1)$, 
$V=\prod_{i=1}^n[(1-\eta)n_i,(1+\eta)n_i]$. In the sequel we consider only
$(m_i)_1^K\in V$; recall that it suffices to sum over such $(m_i)$ in
\eqref{isaac}.	
For more compact notation, write
\begin {equation}
A = \sum_{i=1}^K m_i \qquad \text{and} \qquad B = \sum_{i=1}^K i m_i.
\end {equation}
Note that $A$ and $B$ are $O(N^{1-\alpha})$. Use Stirling's approximation on the first factor in the sum in \eqref{isaac} to get
\begin{align}
& \frac{(N-1-B)!^\alpha}{(N-1-A)!} \notag\\ 
& = \sqrt{\frac{(2\pi(N-1-B))^\alpha}{2\pi(N-1-A)}} \left(\frac{N-1-B}{e}\right)^{\alpha(N-1-B)}  \left(\frac{e}{N-1-A}\right)^{N-1-A} (1+O(N^{-1})) \notag\\ 
& = \sqrt{(2\pi(N-1))^{\alpha -1}}
	\left(\frac{N-1}{e}\right)^{(\alpha-1)(N-1)} (N-1)^{A-\alpha B} \notag\\ 
& \qquad \times ~\exp\Bigg\{\alpha B -A + \alpha\left(N-1-B\right)\log\left(1-\frac{B}{N-1}\right) \notag\\ 
& \qquad\qquad \qquad  - \left(N-1-A\right)\log\left(1-\frac{A}{N-1}\right)\Bigg\} (1+O(N^{-\alpha})) 
\notag\\ & 
= (N-1)!^{\alpha-1} (N-1)^{A-\alpha B}
\exp\Bigg\{\sum_{j=2}^K\frac{\alpha B^{j}  - A^j}{j (j-1)(N-1)^{j-1}}\Bigg\} 
(1+o(1)) 
\notag\\&
= (N-1)!^{\alpha-1} N^{A-\alpha B}
\exp\Bigg\{\sum_{j=2}^K\frac{\alpha B^{j}  - A^j}{j (j-1)N^{j-1}}\Bigg\} 
(1+o(1)),
\label{wchu}
\end{align}
where in the last step we expanded the logarithms and kept only powers of
$A$ and $B$ which contribute for large $N$, and then approximated $N-1$ by $N$.
Hence, \eqref{isaac} yields, using Stirling's formula again,
\begin{equation}\label{zga}
  \begin{split}
\frac{Z_N}{(N-1)!^\ga}
&=\etto
\sum_{(m_i)\in V} 	
\exp\lrcpar{\sum_{j=2}^K\frac{\alpha B^{j}  - A^j}{j (j-1)N^{j-1}}} 
  \prod_{i=1}^{K}\frac{N^{m_i-i\ga m_i} i!^{\alpha m_i}}{m_i!}
\notag\\&
=\sum_{(m_i)\in V} 	\exp\bigpar{f(\mm)+o(1)},
  \end{split}
\raisetag{1.5\baselineskip}
\end{equation}
where
\begin{equation}\label{fdef}
  \begin{split}
  f(\mm)=\sumik \Bigpar{(1-\ga i)m_i\log N + \ga m_i\log(i!)
-m_i\log m_i
\\+m_i
 -\frac12\log(2\pi m_i)}
+
\sum_{j=2}^K\frac{\alpha B^{j}  - A^j}{j (j-1)N^{j-1}}	.
  \end{split}
\end{equation}
Regard $f$ as a function of real variables $\mm$. Then, for $\mm\in V$,
which entails $A,B=O(N^{1-\ga})$,
\begin{equation}\label{f'}
  \begin{split}
  \frac{\partial f}{\partial m_i}
&=
(1-\ga i)\log N + \ga \log (i!)-\log m_i -\frac1{2m_i}
+\sum_{j=2}^K\frac{\alpha i B^{j-1}  - A^{j-1}}{ (j-1)N^{j-1}}	
\\&
=
\log n_i-\log m_i -\frac1{2m_i}
+\frac{\ga i B -A}{N} + O\bigpar{N^{-2\ga}}
\\&
=
\log n_i-\log m_i -\frac1{2m_i}
-\frac{(1-i\ga)m_1}{N} + O\bigpar{N^{-2\ga}}
\\&
= \log n_i-\log m_i + o(1) 
  \end{split}
\raisetag{\baselineskip}
\end{equation}
and, similarly,
\begin{equation} \label{f''}  
  \frac{\partial^2 f}{\partial m_i\partial m_j}
=
-\frac{\gdij}{m_i} +O\Bigparfrac{\gdij}{m_im_j}
+O\Bigparfrac{1}{N}.
\end{equation}

$V$ is compact and $f$ continuous, so $f$ attains its maximum in $V$ at some
point $\nnxx=(\nnx)\in V$.
By \eqref{f'}, for large $N$, 
$\frac{\partial f}{\partial m_i}>0$ when $m_i=(1-\eta)n_i$ and 
$\frac{\partial f}{\partial m_i}<0$ when $m_i=(1+\eta)n_i$, so the maximum
is not attained on the boundary of $V$, i.e.~$|\nx_i-n_i|<\eta n_i$. 
Consequently,
by \eqref{f'},
\begin{equation}\label{f'0}
  0=\frac{\partial f}{\partial m_i}(\nnxx)=\log n_i-\log \nx_i+o(1)
\end{equation}
and thus $\nx_i=\etto n_i$.
A Taylor expansion of $f$ at $\nnxx$ yields, using \eqref{f'0} and \eqref{f''},
for $\mmqq=(\mm)\in V$,
\begin{multline}\label{emma}
f(\mmqq)
=
f(\nnxx)-\frac12\sumik\lrpar{ \frac{(m_i-\nx_i)^2}{\nx_i}
+ O\Bigparfrac{|m_i-\nx_i|^2+|m_i-\nx_i|^3}{n_i^2}}
\\
+O\Bigparfrac{|\mmqq-\nnxx|^2}{N}  
\end{multline}
Choosing $\eta$ small enough, this implies first (for large $N$)
\begin{equation}
  f(\mmqq)
\le
f(\nnxx)-\frac13\sumik \frac{(m_i-\nx_i)^2}{\nx_i},
\end{equation}
which implies that it suffices to consider terms in \eqref{zga} with, say,
$|m_i-\nx_i|<n_i^{1/2} \log N$; let $V_1\subset V$ be the set of such
$\mmqq$. 
For such terms, \eqref{emma} yields
\begin{equation}
  f(\mmqq)
=
f(\nnxx)-\frac12\sumik \frac{(m_i-\nx_i)^2}{\nx_i}+o(1),
\end{equation}
and thus by \eqref{zga}, 
letting $\gb=f(\nnxx)$ be the maximum value of $f$ on $V$,
\begin{equation}\label{magn}
  \frac{Z_N}{(N-1)!^\ga}
=\etto
\sum_{(m_i)\in V_1} 	
\exp\Bigpar{\gb-\frac12\sumik \frac{(m_i-\nx_i)^2}{\nx_i}+o(1)}.
\end{equation}
Since each term here corresponds to the case $X_{i+1,N}=m_i$, $i=1,\dots,K$,
and $\nx_i=\etto n_i$, \eqref{th4gauss} follows.
Furthermore, \eqref{magn} also yields the Poisson approximation result in
Remark \ref{Rth4}, since the Poisson probabilities $\P(Y_{i,N}=m_i,\forall i)$
can easily be approximated by the same Gaussian as in \eqref{magn}; we omit the
details.  

In order to obtain more precise estimates of $\nx_i$, we go back to
\eqref{f'} and refine \eqref{f'0} to
\begin{equation}
  0=\frac{\partial f}{\partial m_i}(\nnxx)=\log n_i-\log \nx_i
-\frac{(1-i\ga)\nx_1}{N}+O\bigpar{N^{-2\ga}+N^{i\ga-1}}
\end{equation}
which yields
\begin{equation}
\log \frac{\nx_i}{n_i}
=
-\frac{(1-i\ga)\nx_1}{N}+O\bigpar{N^{-2\ga}+N^{i\ga-1}}
\end{equation}
and thus, recalling $\nx_1/N=O(N^{-\ga})$,
\begin{equation}\label{erika}
\frac{\nx_i}{n_i}
=
1-\frac{(1-i\ga)\nx_1}{N}+O\bigpar{N^{-2\ga}+N^{i\ga-1}}
\end{equation}
Taking $i=1$ we find $\nx_1/n_1=1+O(N^{-\ga})$, and thus 
$\nx_1-n_1=O(N^{1-2\ga})$, so \eqref{erika} yields 
\begin{equation}\label{matt}
\frac{\nx_i}{n_i}
=
1-\frac{(1-i\ga)n_1}{N}+O\bigpar{N^{-2\ga}+N^{i\ga-1}},
\end{equation}
establishing \eqref{thx1}.

We obtain \eqref{thx2a}--\eqref{thx4b} from \eqref{th4gauss} and
\eqref{thx1}
by checking that in each case,
the omitted terms in the numerator are of smaller order
than the denominator.
\end{proof}

Finally, to evaluate the partition function, we approximate the sum in
\eqref{magn} by a Gaussian integral and obtain
\begin{equation}\label{magn2}
  \frac{Z_N}{(N-1)!^\ga}
=\etto e^{\gb}\prodik \sqrt{2\pi \nx_i}
= e^{\gb+o(1)}\prodik \sqrt{2\pi n_i}.
\end{equation}
We have $\gb=f(\nnxx)$. 
Further, \eqref{matt} shows 
$\nx_i-n_i=O\bigpar{n_i N^{-\ga}}=O\bigpar{N^{1-2\ga}}$, and it follows from
\eqref{emma} that, with $\nnqq=(\nnq)$, 
\begin{equation}\label{sjw}
  f(\nnqq)=f(\nnxx)+O\bigpar{N^{1-3\ga}}
=\gb+O\bigpar{N^{1-3\ga}},
\end{equation}
so it remains to evaluate $f(\nnqq)$.
For $\mmqq=\nnqq$, the final sum in
\eqref{fdef} is 
\begin{equation}
\frac{\alpha B^{2}  - A^2}{2N}+ O\bigpar{N^{1-3\ga}}
=
\frac{(\alpha-1) n_1^2}{2N}+ O\bigpar{N^{1-3\ga}},
\end{equation}
and thus, after some cancellations,
\begin{equation}\label{cec}
  f(\nnqq)=\sumik \Bigpar{n_i-\frac12\log(2\pi n_i)}
-\frac{(1-\alpha) n_1^2}{2N}+ O\bigpar{N^{1-3\ga}}.
\end{equation}
Hence, \eqref{magn2} yields, with \eqref{sjw} and \eqref{cec} and
recalling $n_1=N^{1-\ga}$,
\begin{equation}
  \frac{Z_N}{(N-1)!^\ga}
=
 \exp\lrpar{\sumik n_i
-\frac{1-\alpha}2 N^{1-2\ga}+ O\bigpar{N^{1-3\ga}}
+o(1)}.
\end{equation}
We substitute $n_1$ and $n_2$ from \eqref{ni} and drop $n_i$ for $i\ge3$,
which yields
\eqref{rz}.

\begin{ack}
  This research was done while the authors visited NOR\-DITA, Stockholm, 
during the  program
\emph{Random Geometry and Applications}, 2010.
\end{ack}

\newcommand\arxiv[1]{\url{http://arxiv.org/#1}}

\end{document}